\documentclass[11pt,reqno]{amsart}
\usepackage{graphicx}

\usepackage{amscd,amssymb,amsmath,amsthm}
\usepackage{graphicx}
\usepackage{color}
\usepackage{cite}
\newtheorem{theorem}{Theorem}
\newtheorem{definition}{Definition}

\newtheorem{remark}{Remark}

\numberwithin{equation}{section} 

\begin{document}
\begin{center}
\bf \large {Weakly periodic Gibbs measures of the Ising model on the Cayley tree of order five and six}
\end{center}

\begin{center} Nasir Ganikhodjaev \footnote{Department of Computational and Theoretical Sciences, Faculty of Science, International
Islamic University Malaysia P.O. Box, 141, 25710, Kuantan. E-mail: gnasir@iium.edu.my}, Muzaffar Rahmatullaev\footnote{Institute of mathematics at the National University of Uzbekistan,
29, Do'rmon Yo'li str., 100125, Tashkent, Uzbekistan. E-mail: mrahmatullaev@rambler.ru},  Mohd Hirzie Bin Mohd Rodzhan\footnote{Department of Computational and Theoretical Sciences, Faculty of Science, International
Islamic University Malaysia P.O. Box, 141, 25710, Kuantan. E-mail: mohdhirzie@iium.edu.my}.\end{center}

\begin {abstract} For Ising model on the Cayley tree of order five and six we present new weakly periodic (non-periodic) Gibbs
measures corresponding to normal subgroups of indices two in the
group representation of the  Cayley tree.\end {abstract}


\begin{center}
\section {Introduction}
\end{center}

A Gibbs measure is a mathematical idealization of an equilibrium state of a physical system which consists of a very large number of interacting components.
In the language of Probability Theory, a Gibbs measure is simply the distribution of a stochastic process which, instead of being indexed by the time, is parametrized by the sites of a spatial lattice, and has the special feature of admitting prescribed versions of the conditional distributions  with respect to the configurations outside finite regions. The physical phenomenon of phase transition is reflected by the non-uniqueness of the Gibbs measures for considered model.
The Ising model is realistic enough to exhibit this non-uniqueness of Gibbs measures in which a phase transition is predicted by physics. This fact is one of the main reasons for the physical interest in Gibbs measures. The problem of non-uniqueness, and also the converse problem of uniqueness, are central themes of the theory of Gibbs measures.

Let $M(H)$ be the set of all Gibbs measures defined by Hamiltonian $H.$ Note that this set contain translation-invariant Gibbs measures, periodic Gibbs measures and non-periodic Gibbs measures and one can consider the problem of phase transition  in the classes  translation-invariant Gibbs measures, periodic Gibbs measures and non-periodic Gibbs measures respectively. In [1-5]
have been studied the translational invariant Gibbs measures of the Ising model and some its generalization on the Cayley tree.
The papers [6-8] devoted to study  periodic Gibbs measures with period 2 for models with finite radius of interaction.
In [9-12] the authors introduced new class of Gibbs measures, so-called weakly periodic Gibbs measures and proved the existence of such measures
for the Ising model on the Cayley tree. In the papers [1, 12, 13] have been constructed continuum sets of non periodic
Gibbs measures for the Ising model on the Cayley tree. In [9, 10, 14] the authors considered non-periodic weakly periodic Gibbs measures for the Ising
model on the Cayley tree of order $k<5.$  The present paper is continuation of investigations in [14] and in this paper we study weakly periodic Gibbs measures
on the Cayley tree of order five and six.

\section {Basic definitions and formulation of the problem}

Let $\Gamma^k=(V,L), k\geq 1$, be the Cayley tree of order $k$,
i.e. an infinity graph every vertex of which is incident to
exactly $k+1$ edges. Here $V$ is the set of all vertices, $L$ is
the set of all edges of the tree $\Gamma^k.$ It is known that
$\Gamma^k$ can be represented as a non-commutative group $G_k$, which is the free product of
$k+1$ cyclic groups of the second order \cite{11}.

For an arbitrary point $x^0\in V$ we set  $ W_n=\{x\in V|
d(x^0,x)=n\}$, $V_n=\bigcup\limits_{m=0}^nW_m,$ $L_n=\{<x,y>\in L|
x,y\in V_n\}$, where $d(x,y)$ is the distance between the vertices
$x$ and $y$ in the Cayley tree, i.e. the number of edges in the
shortest path joining the vertices $x$ and $y$.
We write $x\prec y$ if the path from $x^0$ to $y$ goes through $x$.
We call the vertex $y$ a direct successor of $x$, if $y\succ x$ and
$x,y$ are nearest neighbours. The set of the direct successors of
$x$ is denoted by $S(x)$, i.e., if $x\in W_n$ , then
$$S(x)=\{y_i\in W_{n+1}| d(x,y_i)=1, i=1,2,\cdots,k\}.$$
Let $\Phi=\{-1,1\}$ and let $\sigma\in \Omega=\Phi^V$ be a
configuration, i.e. $\sigma=\{\sigma(x)\in \Phi: x\in V\}$. For subset $A\subset V$ we denote by $\Omega_A$ the space of all configurations
defined on the set $A$ and taking values in $\Phi$.

We consider the Hamiltonian of the Ising model:
\begin{equation}\label{1}
H(\sigma)=-J\sum _{<x,y>\in L}\sigma(x) \sigma(y),
\end{equation}
where $J\in R$, $\sigma(x)\in \Phi$ and $ <x,y>$ are nearest
neighbors.

For every $n$, we define a measure
$\mu_n$ on $\Omega_{V_n}$ setting
\begin{equation}\label{2}
\mu_{n}(\sigma_n)=Z^{-1}_n \exp\{-\beta H(\sigma_n)+\sum_{x\in
W_n}h_x \sigma (x)\},
\end{equation} where $h_x\in R, x\in V,$
$\beta=\frac{1}{T}$ ($T$ is temperature, $T>0$),
$\sigma_{n}=\{\sigma(x),x\in V_n\}\in \Omega_{V_n}$, $Z^{-1}_n $
is the normalizing factor, and
$$H(\sigma_n)=-J\sum_{<x,y>\in
L_n}\sigma (x) \sigma (y).$$ The compatibility condition for the
measures $\mu_n(\sigma_n), n\geq 1,$ is
\begin{equation}\label{3}
\sum_{\sigma^{(n)}}\mu_n(\sigma_{n-1},\sigma^{(n)})=\mu_{n-1}(\sigma_{n-1}),
\end{equation}
where $\sigma^{(n)}=\{\sigma(x),x\in W_n\}.$

Let $\mu_n, n\geq1$ be a sequence of measures on the sets
$\Omega_{V_n}$ that satisfy compatibility condition (\ref{3}). By
the Kolmogorov theorem, we then have a unique limit measure $\mu$
on $\Omega_{V}=\Omega$ (called the limit Gibbs measure) such that
$$\mu(\sigma_n)=\mu_n(\sigma_n)$$ for every $n=1,2,...$. It is
known that measures (\ref{2}) satisfies the condition (\ref{3}) if
and only if the set  $h=\{h_x, x\in G_k \}$ of quantities
satisfies the condition
\begin{equation}\label{4}h_x=\sum_{y \in S(x)}f(h_y,\theta), \end{equation}
where $S(x)$ is the set of direct successors of the
vertex  $x\in V$ (see \cite{1}, \cite{2},\cite{3}). Here,
$f(x,\theta)=arctanh(\theta \tanh x)$, $\theta=\tanh(J\beta)$,
 $\beta=\frac{1}{T}.$

Let $G_k/\widehat{G}_k=\{H_1,...,H_r\}$  be a factor group, where
$\widehat{G}_k$ is a normal subgroup of index $r\geq 1$.

\begin{definition}\label{per} A set  $h=\{h_x,\, x\in G_k\}$ of quantities
is called $\widehat{G}_k$-\textit{periodic} if $h_{xy}=h_{x}$, for
all $x\in G_k$ and  $y\in \widehat{G}_k$.
\end{definition}

For $ x\in G_k $ we denote by $x_{\downarrow}$ the unique point of
the set $\{y\in G_k:\langle x,y\rangle\}\setminus S(x)$.

\begin{definition}\label{wp} A set of quantities $h=\{h_x,\, x\in G_k\}$ is called
 $\widehat{G}_k$-\textit{weakly periodic}, if
$h_x=h_{ij}$, for any $x\in H_i, x_{\downarrow}\in H_j$.
\end{definition}

We note that the weakly periodic $h$ coincides with an ordinary
periodic one (see Definition \ref{per}) if the quantity $h_x$ is
independent of $x_{\downarrow}$.

\begin{definition}\label{wpg} A Gibbs measure $\mu$ is said to be $\widehat{G}_k$-(weakly)
periodic if it corresponds to the $\widehat{G}_k$-(weakly)
periodic $h$. We call a $G_k$-periodic measure a
translation-invariant measure.
\end{definition}

In this paper, we study weakly periodic Gibbs measures and
demonstrate that such measures exist for the Ising model on a
Cayley tree of order five and six.

\section {Weakly periodic measures}

The level of difficulty in the describing of weakly periodic Gibbs
measures is related to the structure and index of the normal
subgroup relative to which the periodicity condition is imposed.
It is known (see Chapter 1 of \cite{1}) that in the group $G_k$,
there is no normal subgroup of odd index different from one.
Therefore, we consider normal subgroups of even indices. Here, we
restrict ourself to the case of index two.

We describe $\widehat {G}_k$-weakly periodic Gibbs measures for
any normal subgroup $\widehat{G}_k$ of index two. We note (see
Chapter 1 of \cite{1}) that any normal subgroup of index two of
the group $G_k$ has the form
$$H_A=\left\{x\in G_k:\sum\limits_{i\in A}\omega_x(a_i){\rm -even} \right\},$$
where $\emptyset \neq A\subseteq N_k=\{1,2,\dots,k+1\}$, and
$\omega_x(a_i)$ is the number of letters $a_i$ in a word $x\in
G_k$.

Let $A\subseteq N_k$ and $H_A$ be the corresponding normal
subgroup of index two. We note that in the case $|A|=k+1$, i.e.,
in the case $A = N_k$, weak periodicity coincides with ordinary
periodicity. Therefore, we consider $A\subset N_k$ such that $A\ne
N_k$. Then, in view of (\ref{4}), the $H_A$-weakly periodic set of
$h$ has the form
\begin{equation}\label{5}
h_x=\left\{%
\begin{array}{ll}
    h_{1}, & {x \in H_A, \ x_{\downarrow} \in H_A}, \\[2mm]
    h_{2}, & {x \in H_A, \ x_{\downarrow} \in G_k \backslash H_A}, \\[2mm]
    h_{3}, & {x \in G_k \backslash H_A, \ x_{\downarrow} \in H_A}, \\[2mm]
    h_{4}, & { x \in G_k \backslash H_A, x_{\downarrow}  \in G_k \backslash H_A,}
\end{array}%
\right.\end{equation}
  where $h_{i},i=1,2,3,4$, satisfy the
following equations:
\begin{equation}\label{6}
\left\{%
\begin{array}{ll}
    h_{1}=|A|f(h_{3},\theta)+(k-|A|)f(h_{1},\theta),\\[2mm]
    h_{2}=(|A|-1)f(h_{3},\theta)+(k+1-|A|)f(h_{1},\theta),\\[2mm]
    h_{3}=(|A|-1)f(h_{2},\theta)+(k+1-|A|)f(h_{4},\theta),\\[2mm]
    h_{4}=|A|f(h_{2},\theta)+(k-|A|)f(h_{4},\theta).
\end{array}%
\right.\end{equation}

Consider operator $W:R^4\rightarrow R^4$, defined by

\begin{equation}\label{7}
\left\{%
\begin{array}{ll}
    h'_{1}=|A|f(h_{3},\theta)+(k-|A|)f(h_{1},\theta) \\
    h'_{2}=(|A|-1)f(h_{3},\theta)+(k+1-|A|)f(h_{1},\theta) \\
    h'_{3}=(|A|-1)f(h_{2},\theta)+(k+1-|A|)f(h_{4},\theta) \\
    h'_{4}=|A|f(h_{2},\theta)+(k-|A|)f(h_{4},\theta).\\
\end{array}%
\right.\end{equation}

Note that the system of equations (\ref{6}) describes fixed points of the operator $W,$ i.e. $h=W(h)$.

It is obvious that the following sets are invariant with respect
to operator $W$:
$$ I_1 =\{h\in R^4: h_1=h_2=h_3=h_4\}, \ \ I_2 =\{h\in R^4:
h_1=h_4; h_2=h_3\},$$
$$
I_3 =\{h\in R^4: h_1=-h_4; h_2=-h_3\}.
$$

In \cite{7} it was proved that the system of equation (\ref{6}), on the invariant set $I_2$ has the solutions which belong to $I_1$. The system equation (\ref{6}) on the invariant set $I_1$ reduced to the following equation
\begin{equation}\label{A1}
h=kf(h, \theta).
\end{equation}
The solutions of (\ref{A1}) correspond to translation-invariant Gibbs measures. In this paper we will study weakly periodic (non-periodic, in particular non translation-invariant) Gibbs measures, i.e. we will investigate the fixed points of operator $W$ in the  invariant set $I_3$.

Let $\alpha=\frac{1-\theta}{1+\theta}.$ In \cite{7'} was proven the following statement.
\begin{theorem}\label{th1} Let $|A| =k,\alpha >1$.
\begin{enumerate}
 \item[1)] For $k\leq 3$ all $H_A$-weakly periodic Gibbs measures on $I_3$ are translational invariant.

\item[2)] For $k=4$ there exists a critical value $\alpha_{cr}(\approx�6.3716)$ such that for $\alpha<\alpha_{cr}$ on $I_3$ there exists
one $H_A$-weakly periodic Gibbs measure; for $\alpha=\alpha_{cr}$ on $I_3$ there exist three $H_A$-weakly periodic
Gibbs measures; for $\alpha>\alpha_{cr}$ on $I_3$ there exist five $H_A$-weakly periodic Gibbs measures.
 \end{enumerate}
 \end{theorem}

\begin{remark} Note that one of the  measures described in item 2) of Theorem \ref{th1}, is translation-invariant, but
the other measures are $H_A$-weakly periodic (non-periodic) and differ from measures considered in [9, 10].
\end{remark}

In Theorem \ref{th1} have been considered the cases with $k\leq 4$ (see \cite{7'}).  In this paper we consider the cases with $k\geq5$.

Using the fact that
$$f(h,\theta)={\rm arctanh}(\theta \tanh h)=\frac{1}{2}\ln\frac{(1+\theta)e^{2h}+(1-\theta)}{(1-\theta)e^{2h}+(1+\theta)},$$
and introducing the variables  $z_i=e^{2h_i} \  \ i=1,2,3,4$ one can transform the system of equations (\ref{6}) to the following:

\begin{equation}\label{8}
\left\{%
\begin{array}{ll}
    z_{1}=(\frac{z_3+\alpha}{\alpha z_3+1})^{|A|}\cdot(\frac{z_1+\alpha}{\alpha z_1+1})^{(k-|A|)} \\
    z_{2}=(\frac{z_3+\alpha}{\alpha z_3+1})^{|A|-1}\cdot(\frac{z_1+\alpha}{\alpha z_1+1})^{(k+1-|A|)}\\
    z_{3}=(\frac{z_2+\alpha}{\alpha z_2+1})^{|A|-1}\cdot(\frac{z_4+\alpha}{\alpha z_4+1})^{(k+1-|A|)} \\
    z_{4}=(\frac{z_2+\alpha}{\alpha z_2+1})^{|A|}\cdot(\frac{z_4+\alpha}{\alpha z_4+1})^{(k-|A|)}.\\
\end{array}%
\right.\end{equation}

\begin{theorem}\label{th2} Let $|A| =k.$ Then for arbitrary $k$
the number of $H_A$-weakly periodic (non-periodic) Gibbs measures which correspond to fixed points of operator $W$ on the invariant set $I_3$ does not exceed  four.
 \end{theorem}

 \begin{proof}

Let $|A|=k$.  Then the system of equations
(\ref{8}) has the form
\begin{equation}\label{9}
\left\{%
\begin{array}{ll}
    z_{1}=\left(f(z_3)\right)^{k}\\
    z_{2}=\left(f(z_3)\right)^{k-1}\cdot \left(f(z_1)\right)\\
    z_{3}=\left(f(z_2)\right)^{k-1}\cdot \left(f(z_4)\right)\\
    z_{4}=\left(f(z_2)\right)^{k},
\end{array}%
\right.\end{equation}
where
$f(x)=\frac{x+\alpha}{\alpha x+1}$.
The system of equations (\ref{9}) on the
invariant set $I_3$ has the following form:
\begin{equation}\label{10}
\left\{%
\begin{array}{ll}
    z_{1}=\left(f(\frac{1}{z_2})\right)^{k}\\
    z_{2}=\left(f(\frac{1}{z_2})\right)^{k-1}\cdot \left(f(z_1)\right)
\end{array}%
\right.\end{equation}

and it can be transformed to the following
equation
\begin{equation}\label{11}
z_2=\left(\frac{1+\alpha
z_2}{\alpha+z_2}\right)^{k-1}\frac{\alpha(\alpha+z_2)^k+(1+\alpha
z_2)^k}{(\alpha+z_2)^k+\alpha(1+\alpha z_2)^k}.\end{equation}
Assuming $u=f(z_2)$ we reduce the equation
(\ref{11}) to the equation
\begin{equation}\label{12}
u^{2k}-\alpha u^{2k-1}+\alpha^2 u^{k+1}-\alpha^2 u^{k-1}+\alpha
u-1=0. \end{equation}

According Descartes' rule of signs (see for example \cite{Pra}) the number of positive roots of the polynomial (\ref{12}) is either equal to the number of sign differences between consecutive nonzero coefficients, or is less than it by an even number. Therefore, the equation (\ref{12}) has at most five positive solutions. It is easy to verify that this equation (\ref{12}) is factorized as follows:
\begin{equation}\label{13}
(u^2-1)P_{2k-2}(u)=0, \end{equation}
where $P_{2k-2}(u)$ is a
polynomial of degree $2k-2$. Since one of the roots of (\ref{12}) is $u=1$ which corresponds to translational-invariant Gibbs measure,
 the number of $H_A$-weakly periodic (non-periodic) Gibbs measures does not exceed of four.
\end{proof}
\begin{remark}
In general, the total number of $H_A$-weakly periodic (non-periodic) Gibbs measures (considered everywhere, not only on the invariant set $I_3$) may be greater than four.
 \end{remark}
Recall that a polynomial  $P=\sum_{i=0}^n a_i x^i$ of degree $n$ , is called palindromic (antipalindromic)  if $a_i = a_{n - i}$ (respectively  $a_i = -a_{n - i}$ ) for $i = 0, 1,\cdots, n.$
Note that the polynomial (13) is antipalindromic. It is known that if antipalindromic polynomial of even degree is a multiple of $x^2-1$ (it has -1 and 1 as a roots)
 then its quotient by $x^2-1$  is palindromic (see for example  \cite{Pra}).

\begin{theorem}\label{th3} Let $|A|=k, k=5.$ \\
For the weakly periodic Gibbs measures
corresponding to the set of quantities from  $I_3$ there exists a
critical value $\alpha_{cr}(\approx 2,65)$ such that there is not
any $H_A-$ weakly periodic (nonperiodic)  Gibbs
measures for $0<\alpha<\alpha_{cr}$, there are two $H_A-$ weakly
periodic (nonperiodic)  Gibbs measures for $\alpha=\alpha_{cr}$,
 and there are four $H_A-$ weakly periodic (nonperiodic) Gibbs measures for $\alpha_{cr}<\alpha$.
 \end{theorem}

\begin{proof}
Let $k = 5$.

In this case equation (\ref{12}) has the form
\begin{equation}\label{14}
u^{10}-\alpha u^{9}+\alpha^2 u^{6}-\alpha^2 u^{4}+\alpha
u-1=0. \end{equation}

Now equation (\ref{13}) has the
form

\begin{equation}\label{15}
(u^2-1)\left(u^8-\alpha u^7+ u^6-\alpha
u^5+(\alpha^2+1) u^4-\alpha u^3+ u^2-\alpha
u+1\right)=0. \end{equation}
 From (\ref{15}) we have $ u^2-1=0$
or
\begin{equation}\label{16}
u^8-\alpha u^7+ u^6-\alpha
u^5+(\alpha^2+1) u^4-\alpha u^3+ u^2-\alpha
u+1=0.
\end{equation}
Since $u>0$,  we have that  $u = 1$ is the solution of equation (\ref{15}). We assume that $u\ne 1$. Setting $\xi= u + {1\over u} >
2$, from (\ref{16}) we obtain the equation
\begin{equation}\label{17}
 \xi^4 -\alpha \xi^3-3\xi^2+2\alpha\xi+\alpha^2+1=0. \end{equation}
 The equation (\ref{17}) has at most two positive
 solutions.
 From equation (\ref{17}) we find the parameter $\alpha$:

\begin{equation}\label{18}
\alpha_1=\frac{\xi^3-2\xi+\sqrt{\xi^6-8\xi^4+16\xi^2-4}}{2}:=\gamma_1(\xi),
\end{equation}
\begin{equation}\label{19}
\alpha_2=\frac{\xi^3-2\xi-\sqrt{\xi^6-8\xi^4+16\xi^2-4}}{2}:=\gamma_2(\xi).
\end{equation}

Assume $v=\xi^2$ and $\varphi(v)=v^3-8v^2+16v-4$. We consider
$\varphi'(v)=3v^2-16v+16.$ It is clear that $\varphi'(v)>0$ for
$v>4.$ On the other hand $\varphi(4)<0, \ \ \varphi(+\infty)>0.$
It follows that $\varphi(v)=0$ has a unique solution $v_0$ for
$v>4.$  Therefore, the system of inequalities
$$ \left\{%
\begin{array}{ll}
     \xi^6-8\xi^4+16\xi^2-4\geq 0\\
     \xi>2, \\
\end{array}%
\right.
$$
valid  for $\xi \in [\xi_0, +\infty),$ where
$\xi_0=\sqrt{v_0}(\approx 2.214)$.

Note that $\gamma_1(\xi_0)=\gamma_2(\xi_0)$.

One can check that \begin{equation}\label{20}\lim_{\xi\rightarrow
+\infty} \gamma_{i}(\xi)=+\infty, i=1,2.\end{equation}

It is clear that the function $\gamma_1(\xi)$ is increasing on the $[\xi_0, +\infty)$. Then we get following: for  $\alpha \in (0, \gamma_1(\xi_0))$ there is not $\xi>2$ satisfying the equation (\ref{17}); if $\alpha\in[\gamma_1(\xi_0),+\infty)$ then there exists a unique $\xi>2$ which satisfying the equation (\ref{17}).

Note that if $\xi\in [\xi_0, +\infty)$ then the equation
$\gamma'_2(\xi)=0$ has a unique solution, which is $\xi_1\approx
2.3841,$ and also we get $\gamma_2(\xi_0)\approx 3.21$,
$\gamma_2(\xi_1)\approx 2.65$. Denote $\alpha_{cr}=\gamma_2(\xi_1)$.

Hence it is evident that the function $\gamma_2(\xi)$ reaches its
minimum in $[\xi_0, +\infty)$ at $\xi_1$. Consequently, for
$\alpha \in (0, \gamma_2(\xi_1))$ there is not $\xi>2$ satisfying the equation (\ref{17}), for $\alpha \in \{\gamma_2(\xi_1)\}\cup(\gamma_2(\xi_0); +\infty)$ there exists a unique $\xi>2$ satisfying the equation (\ref{17}), if
$\alpha\in(\gamma_2(\xi_1),\gamma_2(\xi_0))$ then there exist two
$\xi>2$ satisfying the equation (\ref{18}).

Let $n_\alpha$ be the number of solutions of the equation
(\ref{17}). Then $n_\alpha$ has the following form
$$n_\alpha=\left\{%
\begin{array}{ll}
    0, \ \  \mbox{if} \ \ {\alpha\in(0, \alpha_{cr})} \\
    1,  \ \ \mbox{if} \ \  {\alpha=\alpha_{cr}} \\
    2,  \ \ \mbox{if} \ \  {\alpha\in(\alpha_{cr}, +\infty).} \\
\end{array}%
\right.      $$

For $\alpha\in (0, \alpha_{cr})$ from $u+\frac{1}{u}=\xi$ we get
four solutions of the equation (\ref{16}). In this case the
equation (\ref{15}) has five solutions. For $\alpha=\alpha_{cr}$
from $u+\frac{1}{u}=\xi$ we get that the equation (\ref{16}) has
two solutions. Consequently equation (\ref{15}) has three
solutions. In the case $\alpha>\alpha_{cr}$ the equation
(\ref{15}) has a unique solution $u=1$.
\end{proof}

\begin{theorem}\label{th4} Let $|A|=k, k=6.$ \\
For the weakly periodic Gibbs measures
corresponding to the set of quantities from  $I_3$ there exists a
critical value $\alpha_{c}(\approx 1,89)$ such that there is not
any $H_A-$ weakly periodic (nonperiodic)  Gibbs
measures for $\alpha\in(0, \alpha_{c})$, there are two $H_A-$ weakly
periodic (nonperiodic)  Gibbs measures for $\alpha\in[2, 3]\cup\{\alpha_{c}\}$,
 and there are four $H_A-$ weakly periodic (nonperiodic) Gibbs measures for $\alpha\in(\alpha_{c}, 2)\cup(3,+\infty)$.
 \end{theorem}
\begin{proof}
Let $k=6$.

In this case (\ref{12}) has the form
\begin{equation}\label{20}
u^{12}-\alpha u^{11}+\alpha^2 u^{7}-\alpha^2 u^{5}+\alpha
u-1=0. \end{equation}

The function $y:=y(u)=u^{12}-\alpha u^{11}+\alpha^2 u^{7}-\alpha^2 u^{5}+\alpha
u-1$ with $\alpha=4.1$ is plotted in Fig. 1.

\begin{figure}
  \includegraphics[width=0.8\textwidth]{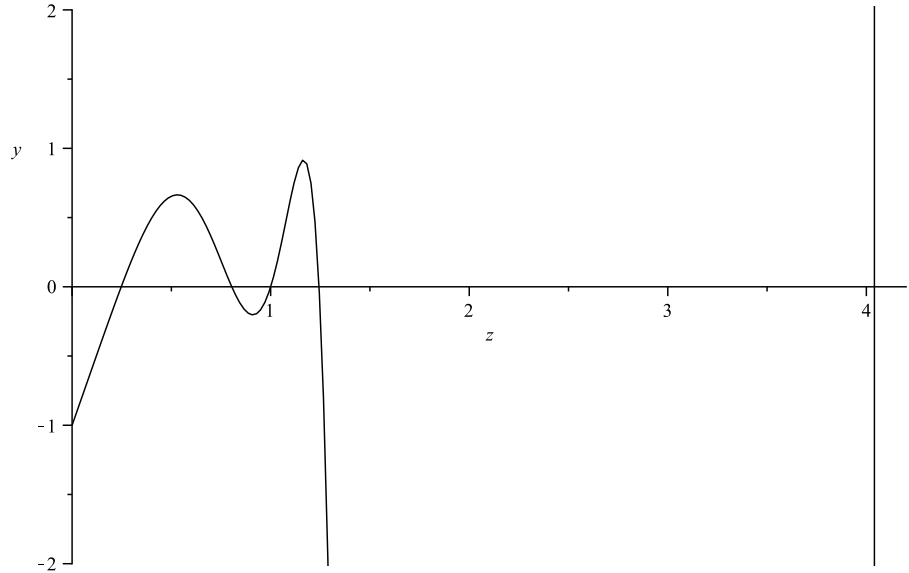}\\
  \caption{\textit{The function $y=y(u)$ with $\alpha=4.1$.}}\label{1}
\end{figure}

In this case, the equation (\ref{13}) has the
form

\begin{equation}\label{21}
(u^2-1)\left(u^{10}-\alpha u^9+ u^8-\alpha
u^7+u^6+(\alpha^2-\alpha) u^5+u^4-\alpha u^3+ u^2-\alpha
u+1\right)=0. \end{equation}
 From (\ref{21}) we have $ u^2-1=0$
or
\begin{equation}\label{22}
u^{10}-\alpha u^9+ u^8-\alpha
u^7+u^6+(\alpha^2-\alpha) u^5+u^4-\alpha u^3+ u^2-\alpha
u+1=0.
\end{equation}
Since $u>0$, we have that  $u = 1$ is the solution of equation (\ref{21}). We assume that $u\ne 1$. Setting $\xi= u + {1\over u} >
2$, from (\ref{22}) we obtain the equation
\begin{equation}\label{23}
 \xi^5 -\alpha \xi^4-4\xi^3+3\alpha\xi^2+3\xi+\alpha^2-\alpha=0. \end{equation}

 From equation (\ref{23}) we find the parameter $\alpha$:

\begin{equation}\label{24}
\alpha_1=\frac{\xi^4-3\xi^2+1-\sqrt{\xi(\xi^2-1)(\xi^2-3)(\xi-2)(\xi^2+2\xi+2)+1}}{2}:=\alpha_1(\xi),
\end{equation}
\begin{equation}\label{25}
\alpha_2=\frac{\xi^4-3\xi^2+1+\sqrt{\xi(\xi^2-1)(\xi^2-3)(\xi-2)(\xi^2+2\xi+2)+1}}{2}:=\alpha_2(\xi).
\end{equation}

One can check that \begin{equation*}\lim_{\xi\rightarrow
+\infty} \alpha_{i}(\xi)=+\infty, i=1,2,\end{equation*}
and $\xi(\xi^2-1)(\xi^2-3)(\xi-2)(\xi^2+2\xi+2)+1$ is positive for all $\xi \geq 2.$

Note that if $\xi\in [2, +\infty)$ the the equation $\alpha_1'(\xi)=0$ has a unique solution which is $\xi_0\approx 2,077$, and also we get $\alpha_1(\xi_0)\approx 1,89$. Hence it is clear that the function $\alpha_1(\xi)$ reaches its
minimum in $[2, +\infty)$ at $\xi_0$. Consequently, for
$\alpha \in (0, \alpha_1(\xi_0))$ there is not $\xi\geq 2$ satisfying the equation (\ref{23}), for $\alpha \in \{\alpha_1(\xi_0)\}\cup(\alpha_1(2); +\infty)$ there exist unique $\xi\geq 2$ satisfying the equation (\ref{23}),
$\alpha\in(\alpha_1(\xi_0),\alpha_1(2)]$ there exist two
$\xi\geq 2$ satisfying the equation (\ref{23}).

It is clear that the function $\alpha_2(\xi)$ is increasing on the $[2, +\infty)$. Then we get following: for  $\alpha \in (0, \alpha_2(2))$ there is not $\xi>2$ satisfying the equation (\ref{23}); $\alpha\in[\alpha_2(2),+\infty)$ there exist a unique $\xi>2$ which satisfying the equation (\ref{23}).

Denote $\alpha_1(\xi_0)=\alpha_c$.

Let $n_\alpha$ be the number of solutions of the equation
(\ref{22}). Then $n_\alpha$ has the following form
$$n_\alpha=\left\{%
\begin{array}{lll}
    0, \ \  \mbox{if} \ \ {\alpha\in(0, \alpha_{c})} \\
    1,  \ \ \mbox{if} \ \  {\alpha\in(\alpha_1(2), \alpha_2(2))\cup\{\alpha_{c}\}} \\
    2,  \ \ \mbox{if} \ \  {\alpha\in(\alpha_{c}, \alpha_1(2)]\cup[\alpha_2(2),+\infty).} \\
\end{array}%
\right.$$

For $\alpha\in (0, \alpha_{c})$  equation (\ref{21}) has a unique solution $u=1$. For $\alpha \in (2, 3)\cup \{\alpha_c\}$
from $u+\frac{1}{u}=\xi$ we get
two solutions of the equation (\ref{22}). In this case the
equation (\ref{21}) has five solutions. Note that $\alpha_1(2)=2$ and $\alpha_2(2)=3$ and in the case $\xi=2$ from $u+\frac{1}{u}=\xi$ we get $u=1.$ Consequently, for $\alpha=2$ and $\alpha=3$ qe get two solution of equation (\ref{22}) which different from $u=1.$ In this case equation (\ref{21}) has three solutions. For $\alpha \in (\alpha_c, \alpha_1(2))\cup (\alpha_2(2), +\infty)$ from $u+\frac{1}{u}=\xi$ we get that the equation (\ref{22}) has
two solutions. In this the equation
(\ref{21}) has a five solutions.

Let $N_\alpha$ be the number of solutions of the equation
(\ref{21}). Then $N_\alpha$ has the following form
$$N_\alpha=\left\{%
\begin{array}{lll}
    1, \ \  \mbox{if} \ \ {\alpha\in(0, \alpha_{c})} \\
    3,  \ \ \mbox{if} \ \  {\alpha\in[2, 3]\cup\{\alpha_{c}\}} \\
    5,  \ \ \mbox{if} \ \  {\alpha\in(\alpha_{c}, 2)\cup(3,+\infty).} \\
\end{array}%
\right.$$

\end{proof}

\section {DISCUSSION}

It is well known [16-19] the formulation of Ising model according to which every node $x$ of a lattice
corresponds to the two-valued variable $\sigma(x)$ with values $+1$ or $-1$.
If "objects" connected with nodes $x$ and $x'$ are in the same state, then $\sigma(x)\sigma(x')=+1$ ,but if they
are in different state, then $\sigma(x)\sigma(x')=-1$. Clear that in such an interpretation the definition of "object"
connected with the node $x$ can be discussed very widely. For example, it may be a magnetic moment
of an ion in a crystal having two directions [20] or it may be atoms of two kinds in a binary alloy [21]
(value $\sigma(x)=+1$ corresponds to an occupation of $x$-th node by an atom of one kind, and $\sigma(x)=-1$
when occupation take place by an atom of another kind). Others interpretations of the Ising model are
connected with an investigation of adsorption phenomenon on a surface [22], with DNA melting [23],
with the theory of latticed gas [24] and with other questions of theory of change of a phase of type order-
disorder.
For considered model on $\Gamma^5$ we find critical value $\alpha_{cr}(\approx 2,65)$ such that there are two $H_A-$ weakly
periodic (nonperiodic)  Gibbs measures for $\alpha=\alpha_{cr}$,
 and there are four $H_A-$ weakly periodic (non-periodic) Gibbs measures for $\alpha_{cr}<\alpha$.

On $\Gamma^6$ we find critical value $\alpha_{c}(\approx 1,89)$ such that there are two $H_A-$ weakly
periodic (non-periodic)  Gibbs measures for $\alpha\in[2, 3]\cup\{\alpha_{c}\}$,
 and there are four $H_A-$ weakly periodic (non-periodic) Gibbs measures for $\alpha\in(\alpha_{c}, 2)\cup(3,+\infty).$.

\textbf{Aknowledgments.} This research was supported by Ministry of Higher Education Malaysia (MOHE) under grant FRGS 14-116-0357.
Second author Rahmatullaev Muzaffar thanks IIUM for providing financial
support (grant FRGS 14-116-0357) and all facilities.

\end{document}